\documentclass[a4paper,11pt]{article}
\usepackage[top=25truemm,bottom=25truemm,left=15truemm,right=15truemm]{geometry}
\usepackage{amsmath,amssymb,amscd,amsthm}
\usepackage{bm}
\usepackage{multicol}

\usepackage[dvipdfmx]{color}
\usepackage{latexsym}
\usepackage{mathtools}
\usepackage{comment}
\usepackage[dvipdfmx]{xcolor}
\usepackage{cite}
\usepackage{float}
\usepackage{here}
\usepackage[dvipdfmx]{graphicx}

\newcommand{\Hline}[1]{\noalign{\hrule height #1}}   

\begin{document}
\title{Necessary and Sufficient Conditions for\\ Capacity-Achieving Private Information Retrieval \\with Non-Colluding and Colluding Servers}
\author{Atsushi Miki, Yusuke Morishita, Toshiyasu Matsushima}
\date{}
\maketitle

\begin{abstract}
  Private information retrieval (PIR) is a mechanism for efficiently downloading messages 
  while keeping the index secret. 
  Here, PIR schemes in which servers do not communicate with each other are called standard PIR schemes, 
  and PIR schemes in which some servers communicate with each other are called colluding PIR schemes. 
  The information-theoretic upper bound on efficiency has been derived in previous studies. 
  However, the conditions for PIR schemes to keep privacy, to decode the desired message, and to achieve that upper bound have not been clarified in matrix form. 
  In this paper, we prove the necessary and sufficient conditions for the properties of standard PIR and colluding PIR.
  Further, we represent the properties in matrix form. 
\end{abstract}

\section{Introduction}
Private information retrieval (PIR) is used as a mechanism to retrieve a message from database servers through efficient communication while keeping the index of the desired message secret to the servers.
The communication cost is defined by the entropy of sending a query (hereinafter, upload) and receiving a response (hereinafter, download).
The simplest way to keep privacy is to download the entire library from a single server, but the communication efficiency is low. 
To improve efficiency, a user sends an encoded query to multiple database servers and decodes their responses.

PIR schemes in which servers do not communicate with each other are called standard PIR schemes, and PIR schemes in which some servers communicate with each other are called colluding PIR schemes. 
PIR schemes are evaluated mainly by the download efficiency which is called rate.
Therefore, PIR schemes that achieve the information-theoretic upper bound of the rate (hereinafter, capacity) are desired. 

In other words, PIR has three properties: privacy to protect the user's choice, correctness to decode the desired message, and capacity achievability.

For correctness and privacy, various methods were proposed, and the capacities of standard PIR schemes \cite{7889028,sun2018capacity,8720234} and that ofcolluding PIR schemes \cite{8119895} were derived from previous studies. 
Also, some methods for the PIR schemes that achieve the capacity were proposed in some previous studies \cite{8119895,8457293,9145637,abs-1710-01190,9023400,8598994,8509632,8437880,8552404}.
However, how to satisfy these properties is not clarified as a query matrix propoerties.

In this study, we prove the necessary and sufficient conditions for privacy, correctness, and capacity achievability.
The conditions simplify the construction methods of PIR schemes that achieve capacity and make it possible to verify that methods have privacy and correctness.

The remainder of this paper is organized as follows: 
Chapter 2 sets up the problem and provides an overview of PIR.
In Chapter 3, we derive the conditions for standard PIR.
Chapter 4 presents an example of a capacity-achieving PIR and shows that the conditions are satisfied.
In Chapter 5, we derive the conditions for colluding PIR.

\section*{Notation}
For a fixed positive integer $Z$, $[1:Z] \coloneq \{1, 2, \dots , Z\}$.  
For a vector $\mathcal{V}=(v[1],\dots,v[Z])$ and a set $\mathbf{s}=\{s_1,\dots,s_n\} \subset [1:Z],\ n\in \mathbb{N}$, $\mathcal{V}[{\mathbf{s}}]$ denotes $(v[s_1],\dots,v[s_n])$.
For a $r \times c$ matrix $\mathcal{M}$, its row vectors are denoted by $\mathcal{M}_1,\dots,\mathcal{M}_r$, 
and for sets $\mathbf{t}=\{t_1,\dots,t_i\} \subset [1:r], \mathbf{u}=\{u_1,\dots,u_j\} \subset [1:c],\ i,j\in \mathbb{N}$,
$\mathcal{M}_{\mathbf{t}}$ denotes 
$\begin{pmatrix}
  \mathcal{M}_{t_1}\\
    \vdots \\
  \mathcal{M}_{t_i} 
\end{pmatrix}$, 
and 
$\mathcal{M}_{\mathbf{t}}[\mathbf{u}]$ denotes 
$\begin{pmatrix}
  \mathcal{M}_{t_1}[\mathbf{u}]\\
    \vdots \\
  \mathcal{M}_{t_i}[\mathbf{u}] 
\end{pmatrix}$.
We denote by $A_{n_1:n_2}, n_1, n_2 \in \mathbb{Z}$
the set $\{A_{n_1},A_{n_1+1}, \dots , A_{n_2},\}$, for $n_1 \leq n_2$;
if $n_1 > n_2$, $\mathbf{A}_{n_1:n_2}$ is the empty set $\phi$.
For a finite field that has order $p$, we use the notation $\mathbb{F}_p$.

\section{Problem Statement}
We consider $S$ servers such that each server stores $M$ messages $W_{1:M}$, which are independent vectors of $L_w$ length.
\begin{align}
  &M \geq 2, S\geq 2,
  \\ & H(W_{1:M}) = H(W_1) + \dots + H(W_M),
  \\ & H(W_1) = \dots = H(W_M) = L_w.
\end{align}

A user retrieves a message $W_m\ (m\in [1:M])$
while keeping the selected index $m$ secret from the servers.
The flow of the process is shown in Fig. \ref{PIR_flow}.
The user sends the query $Q_j^{(m)}\ (j\in [1:S])$ 
generated by a random function $g$ to the servers 
and each server generates a response $X_j^{(m)}\ (j\in [1:S])$ 
by a function $h$ respectively and returns it to the user. 
Therefore, let $\theta$ be a random variable that indicates the index chosen by the user, then
for all $m,m^\prime \in [1:M]$, the following equation must be satisfied to keep the index $m$ secret from the servers.
\begin{eqnarray}
  [\mathrm{Privacy}] \quad \mathrm{Pr}\{Q_j^{(\theta)}=q_j | \theta=m\} = \mathrm{Pr}\{Q_j^{(\theta)}=q_j | \theta=m^{\prime}\}.
\end{eqnarray}
Let $\mathcal{T}$ be the index set of $T$ colluding servers,
then for any $m,m^\prime \in [1:M]$, the following equation must be satisfied to keep the index $m$ secret from the colluding servers.
\begin{eqnarray}
 [\mathrm{Colluding\ Privacy}] \quad \mathrm{Pr}\{Q_{\mathcal{T}}^{(\theta)}=q | \theta=m\} = \mathrm{Pr}\{Q_{\mathcal{T}}^{(\theta)}=q | \theta=m^{\prime}\}.
\end{eqnarray}

The user decodes $W_m$ using the responses and the function $d$.
Therefore, the following equation must be satisfied.
\begin{eqnarray}
 [\mathrm{Correctness}] \quad H(W_m|X_{1:S}^{(m)},Q_{1:S}^{(m)})=0.
\end{eqnarray}

\begin{figure}[tbh]
  \begin{center}
  \includegraphics[width=7.5cm]{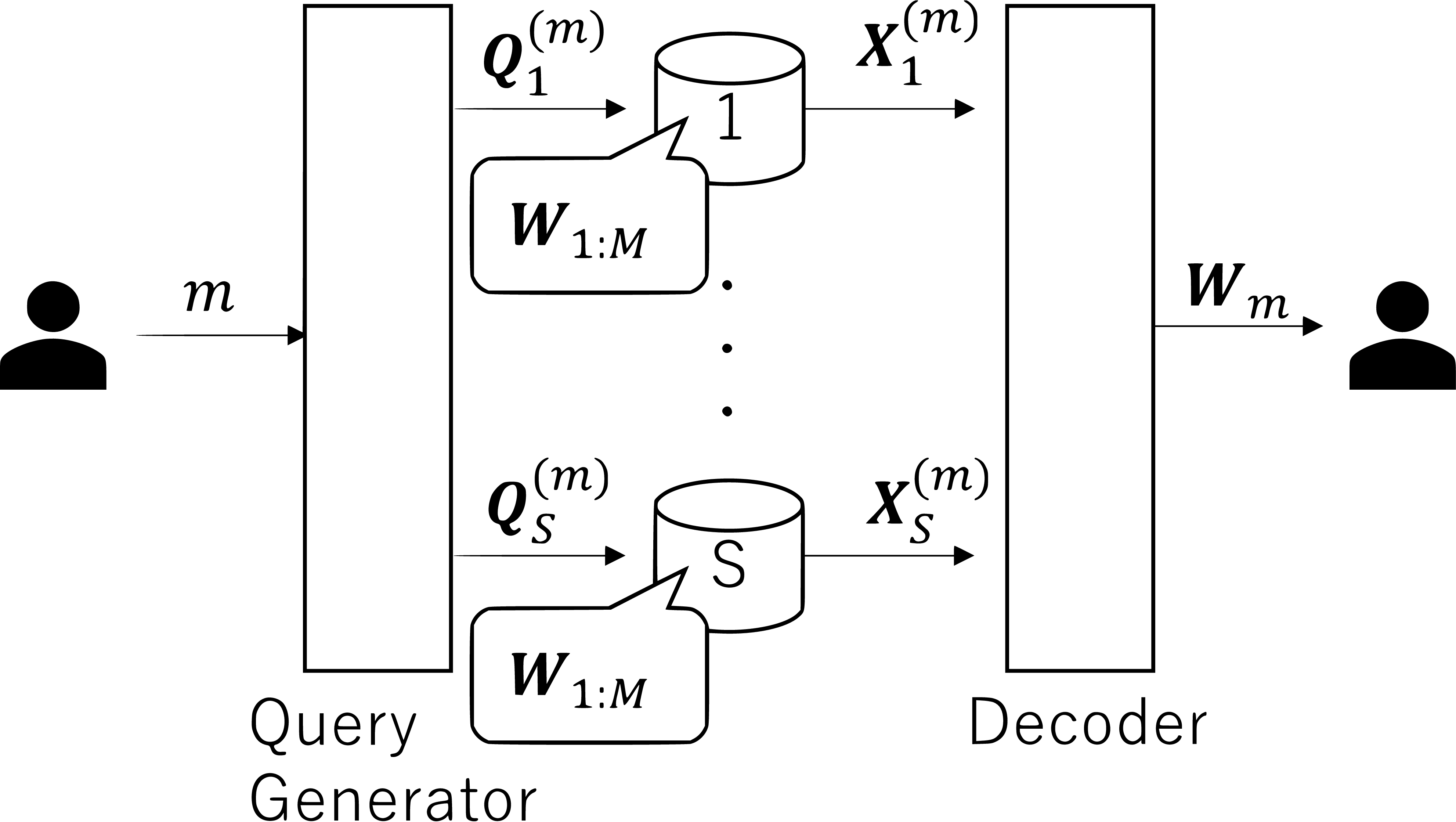}
  \end{center}
  \caption{PIR flow}
  \label{PIR_flow}
\end{figure}  



The rate is defined as follows:
\begin{eqnarray}
  \mathcal{R}=\min_{m\in [1:M]}\frac{H(W_m)}{\sum_{i=1}^S H(X_i^{(m)}|Q_i^{(m)})}.
\end{eqnarray}
The capacity is an information-theoretic upper bound on the rate $\mathcal{R}$, 
and the capacity of the standard PIR was derived from previous studies \cite{7889028} as follows:

\begin{eqnarray}
  [\mathrm{PIR\ Capacity}] \quad \mathcal{R} \leq \frac{1-\frac{1}{S}}{1-(\frac{1}{S})^M}.
 \end{eqnarray}
Let $T$ be the number of colluding servers, the capacity of the colluding PIR was derived as follows \cite{8119895}:
 \begin{eqnarray}
  [\mathrm{Colluding\ PIR\ Capacity}] \quad \mathcal{R} \leq \frac{1-\frac{T}{S}}{1-(\frac{T}{S})^M}.
\end{eqnarray}


\section{Standard PIR}
\subsection{Conditions for Correctness}
Let the queries be $Q_j^{(m)}=(Q_j^{(m)}[1],\dots,Q_j^{(m)}[M])$,
the responses are as follows:
\begin{align}
\begin{pmatrix}
X_1^{(m)} \\ \vdots \\X_S^{(m)} 
\end{pmatrix}
=
\begin{pmatrix}
  Q_1^{(m)}[1] & \dots & Q_1^{(m)}[M]  \\
 \vdots & \ddots & \vdots \\
  Q_S^{(m)}[1] & \dots & Q_S^{(m)}[M]
\end{pmatrix}
\begin{pmatrix}
  W_1 \\ \vdots \\ W_M 
\end{pmatrix}.
\end{align}

Using a decoding matrix $D_m$, the decoding scheme is as follows:
\begin{align}
D_m
\begin{pmatrix}
    X_1^{(m)} \\ \vdots \\X_S^{(m)} 
\end{pmatrix}
=
W_m  .
\end{align}
Therefore, there exists a matrix $D_m$ and the following equation should be satisfied:
\begin{align}
  D_m \begin{pmatrix}Q_1^{(m)}\\\vdots \\ Q_S^{(m)}\end{pmatrix} = 
    \begin{pmatrix}
      \mathbf{O}_{(m-1)L_w \times L_w}\\
      \mathbf{E}\\
      \mathbf{O}_{(M-m)L_w \times L_w}
     \end{pmatrix}^{\top},
\end{align}
where $\mathbf{O}_{l_1 \times l_2}$ denotes the $l_1 \times l_2$ zero matrix and 
$\mathbf{E}$ denotes the $L_w \times L_w$ identity matrix.

\subsection{Conditions for Privacy}
Let $F$ be a random key which is uniformly distributed on a certain finite set $\mathcal{F}$.
Assume that a query set $Q_{1:S}^{(m)}$ is constructed by a bijective function $\psi(m,F)$, 
then $Q_{1:S}^{(\theta)}$ is chosen uniformly from a query realization set $\mathcal{Q}$.
Let $p(q_{1:S},m,f)=Pr\{Q_{1:S}^{(\theta)}=q_{1:S}, \theta=m, F=f\}$ be a probability mass function of the queries, then the following equation holds:
 \begin{align}
  p(q_{1:S}|m, f) = 
  \begin{cases}
    1 \quad \mathrm{if}\ q_{1:S}=\psi(m,f)\\
    0 \quad \mathrm{otherwise} \ .
  \end{cases}\label{lem1}
    \end{align}
 Therefore, the following lemma holds by using (\ref{lem1}):
\newtheorem{lemma}{Lemma}
\begin{lemma}[entropy of index under given queries]
\begin{align}
  H(\theta | Q_{1:S}^{(\theta)}) = 0.
\end{align}
\end{lemma}
\begin{proof}
  \begin{align}
    & H(\theta | Q_{1:S}^{(\theta)}) \notag
    \\ & = -\sum_{q_{1:S} \in \mathcal{Q}} \sum_{m \in [1:M]}  \sum_{f\in \mathcal{F}}p(q_{1:S},m,f)  \log p(m | q_{1:S}, f) \notag
     \\&= 0\label{lem2}.
  \end{align}
\end{proof}
Let $\bar{j} \coloneq [1:S] \setminus \{j\}$.
Further, assume $H(Q_{1:S}^{\theta} | \theta=m)=H(Q_{1:S}^{\theta} | \theta=m^{\prime})$, then the following equations hold: 
    \begin{align}
      &\mathrm{Pr}\{Q_j^{(\theta)}=q_j | \theta=m\} = \mathrm{Pr}\{Q_j^{(\theta)}=q_j | \theta=m^{\prime}\}
      \\& \Rightarrow  H(Q_{j}^{(\theta)} | \theta=m) = H(Q_{j}^{(\theta)} | \theta=m^{\prime})
      \\& \Leftrightarrow  H(Q^{(\theta)}_{1:S}|Q^{(\theta)}_j=q_j, \theta=m) = H(Q^{(\theta)}_{1:S}|Q^{(\theta)}_j=q_j, \theta=m^{\prime})
      \\& \Leftrightarrow  H(Q^{(m)}_{1:S}|Q^{(m)}_j=q_j) = H(Q^{(m^{\prime})}_{1:S}|Q^{(m^{\prime})}_j=q_j). \label{asm}
    \end{align}  
Because of Lemma 1, Theorem 1 holds.
\newtheorem{theorem}{Theorem}
\begin{theorem}[Conditions for Privacy]
\begin{align}
  &\mathrm{Pr}\{Q_j^{(\theta)}=q_j | \theta=m\} = \mathrm{Pr}\{Q_j^{(\theta)}=q_j | \theta=m^{\prime}\}
  \\& \Leftrightarrow |\{Q_{1:S}^{(\theta)} | Q_{j}^{(\theta)} = q_j\}|=M|\{Q_{1:S}^{(m)} | Q_{j}^{(m)} = q_j\}|. \notag
\end{align}
\end{theorem}
\begin{proof}
  \begin{align}
  &\mathrm{Pr}\{Q_j^{(\theta)}=q_j | \theta=m\} = \mathrm{Pr}\{Q_j^{(\theta)}=q_j | \theta=m^{\prime}\}
  \\ &\Leftrightarrow \mathrm{Pr}\{Q_j^{(\theta)}=q_j, \theta=m\} = \mathrm{Pr}\{Q_j^{(\theta)}=q_j | \theta=m^{\prime}\} \mathrm{Pr}\{\theta = m\}
  \\ &\Leftrightarrow \mathrm{Pr}\{Q_j^{(\theta)}=q_j, \theta=m\} = \mathrm{Pr}\{Q_j^{(\theta)}=q_j , \theta=m^{\prime}\} \mathrm{Pr}\{\theta = m\} / \mathrm{Pr}\{\theta=m^{\prime}\}
  \\ &\Leftrightarrow \mathrm{Pr}\{\theta=m | Q_j^{(\theta)}=q_j\} = \mathrm{Pr}\{\theta=m^{\prime} | Q_j^{(\theta)}=q_j\} \mathrm{Pr}\{\theta = m\} / \mathrm{Pr}\{\theta=m^{\prime}\}
  \\ &\Leftrightarrow \mathrm{Pr}\{\theta=m | Q_j^{(\theta)}=q_j\} = \mathrm{Pr}\{\theta=m^{\prime} | Q_j^{(\theta)}=q_j\}
  \\ &\Leftrightarrow \mathrm{Pr}\{\theta=m | Q_j^{(\theta)}=q_j\} = \frac{1}{M} \label{p1}
  \\ &\Leftrightarrow H(\theta | Q_j^{(\theta)}=q_j) = \log M \label{p2}
  \\ &\Leftrightarrow H(\theta | Q_j^{(\theta)}=q_j) - H(\theta | Q_{\bar{j}}^{(\theta)},Q_j^{(\theta)}=q_j)= \log M  \label{p3} 
  \\ &\Leftrightarrow I(\theta;Q_{\bar{j}}^{(\theta)} | Q_j^{(\theta)}=q_j) = \log M
  \\ &\Leftrightarrow H(Q_{\bar{j}}^{(\theta)} | Q_j^{(\theta)}=q_j) - H(Q_{\bar{j}}^{(\theta)} | \theta ,Q_j^{(\theta)}=q_j) = \log M
  \\ &\Leftrightarrow H(Q_{\bar{j}}^{(\theta)} | Q_j^{(\theta)}=q_j) = \log M + H(Q_{\bar{j}}^{(\theta)} | \theta ,Q_j^{(\theta)}=q_j) 
  \\ &\Leftrightarrow H(Q_{\bar{j}}^{(\theta)} | Q_j^{(\theta)}=q_j) = \log M + \sum_{i=1}^M P(\theta=i) H(Q_{\bar{j}}^{(\theta)} | \theta=i ,Q_j^{(\theta)}=q_j) 
  \\ &\Leftrightarrow H(Q_{\bar{j}}^{(\theta)} | Q_j^{(\theta)}=q_j) = \log M + \sum_{i=1}^M \frac{1}{M}H(Q_{\bar{j}}^{(\theta)} | \theta=i ,Q_j^{(\theta)}=q_j) 
  \\ &\Leftrightarrow H(Q_{\bar{j}}^{(\theta)} | Q_j^{(\theta)}=q_j) = \log M + H(Q_{\bar{j}}^{(\theta)} | \theta=m ,Q_j^{(\theta)}=q_j) \label{p4}
  \\ &\Leftrightarrow H(Q_{\bar{j}}^{(\theta)} | Q_j^{(\theta)}=q_j) = \log M + H(Q_{\bar{j}}^{(m)} | Q_j^{(m)}=q_j) \label{ent}
  \\ &\Leftrightarrow |\{Q_{1:S}^{(\theta)} | Q_{j}^{(\theta)} = q_j\}|=M|\{Q_{1:S}^{(m)} | Q_{j}^{(m)} = q_j\}| \label{sym}.
 \end{align}
\end{proof}
(\ref{p1}) is because $\theta$ follows a uniform distribution,
and the entropy is maximized in that case, so (\ref{p2}) holds. 
(\ref{p3}) is because of Lemma 1. 
(\ref{p4}) is because $H(Q_{1:S}^{(i)}|Q^{(i)}_j=q_j)=H(Q_{1:S}^{(i^\prime)}|Q^{(i^\prime)}_j=q_j)$.
(\ref{sym}) is because (\ref{ent}) indicates the number of query realizations is $M$ times larger when $\theta$ is not given than when $\theta = m$. 
\subsection{Conditions for Capacity Achievability}
Let $\pi$ be a random substitution for $[1:M]$.
The capacity of the standard PIR is derived using the following two lemmas \cite{7889028,8720234}:
\begin{lemma}[Interference Lower Bound Lemma]
    \begin{align}
     &\sum_{j=1}^S H(X_j^{(\pi[1])}|Q_{1:S}^{(\pi[m])})-L_w \notag
     \\&\geq I(W_{\pi[2:M]}; Q_{1:S}^{(\pi[1])}, X_{1:S}^{(\pi[1])}|W_{\pi[1]}). 
    \end{align}
\end{lemma}
\begin{proof}
  \begin{align}
    &I(W_{\pi[2:M]}; Q_{1:S}^{(\pi[1])}, X_{1:S}^{(\pi[1])}|W_{\pi[1]})  
    \\&= I(W_{\pi[2:M]}; Q_{1:S}^{(\pi[1])}, X_{1:S}^{(\pi[1])}, W_{\pi[1]})
    \\&= I(W_{\pi[2:M]}; Q_{1:S}^{(\pi[1])}, X_{1:S}^{(\pi[1])})
    \\&= I(W_{\pi[2:M]}; X_{1:S}^{(\pi[1])}|Q_{1:S}^{(\pi[1])})
    \\&= H(X_{1:S}^{(\pi[1])}|Q_{1:S}^{(\pi[1])}) - H(X_{1:S}^{(\pi[1])} | Q_{1:S}^{(\pi[1])}, W_{\pi[2:M]})
    \\&\leq \sum_{j=1}^S H(X_j^{(\pi[1])}|Q_{1:S}^{(\pi[1])}) - H(X_{1:S}^{(\pi[1])} | Q_{1:S}^{(\pi[1])}, W_{\pi[2:M]})
    \\&= \sum_{j=1}^S H(X_j^{(\pi[1])}|Q_{1:S}^{(\pi[1])}) - H(X_{1:S}^{(\pi[1])},W_{\pi[1]} | Q_{1:S}^{(\pi[1])}, W_{\pi[2:M]}) + \underbrace{H(W_{\pi[1]} | Q_{1:S}^{(\pi[1])}, W_{\pi[2:M]}, X_{1:S}^{(\pi[1])})}_{=0}
    \\&= \sum_{j=1}^S H(X_j^{(\pi[1])}|Q_{1:S}^{(\pi[1])}) -  H(W_{\pi[1]} | Q_{1:S}^{(\pi[1])}, W_{\pi[2:M]}) -  \underbrace{H(X_{1:S}^{(\pi[1])} | Q_{1:S}^{(\pi[1])}, W_{1:M})}_{=0}
    \\&= \sum_{j=1}^S H(X_j^{(\pi[1])}|Q_{1:S}^{(\pi[1])})
     -L_w.
  \end{align}
\end{proof}

\begin{lemma}[Indution Lemma]
\begin{align}
  &I(W_{m:M} ; Q_{1:S}^{(\pi[m-1])}, X_{1:S}^{(\pi[m-1])}| W_{1:m-1})
  \\& \geq \frac{1}{S}I(W_{m+1:M} ; Q_{1:S}^{(\pi[m])}, X_{1:S}^{(\pi[m])}| W_{1:m}) + \frac{L_w}{S} \notag .
 \end{align}
\end{lemma}
\begin{proof}
  \begin{align}
    &I(W_{\pi[m:M]} ; Q_{1:S}^{(\pi[m-1])}, X_{1:S}^{(\pi[m-1])}| W_{\pi[1:m-1]})
    \\&\geq \frac{1}{S} \sum_{j=1}^S I(W_{\pi[m:M]} ; Q_{j}^{(\pi[m-1])}, X_{j}^{(\pi[m-1])}| W_{\pi[1:m-1]})
    \\&= \frac{1}{S} \sum_{j=1}^S I(W_{\pi[m:M]} ; Q_{1:S}^{(\pi[m-1])}, X_{j}^{(\pi[m-1])}| W_{\pi[1:m-1]}) \label{ind}
    \\&= \frac{1}{S} \sum_{j=1}^S I(W_{\pi[m:M]} ; Q_{1:S}^{(\pi[m])}, X_{j}^{(\pi[m])}| W_{\pi[1:m-1]}) \label{priv}
    \\&= \frac{1}{S} \sum_{j=1}^S H(X_{j}^{(\pi[m])}| Q_{1:S}^{(\pi[m])}, W_{\pi[1:m-1]}) 
    \\&\geq \frac{1}{S} \sum_{j=1}^S H(X_{j}^{(\pi[m])}| Q_{1:S}^{(\pi[m])}, W_{\pi[1:m-1]},X_{1:j-1}^{(\pi[m])})
    \\&= \frac{1}{S} \sum_{j=1}^S H(X_{j}^{(\pi[m])}| Q_{1:S}^{(\pi[m])}, W_{\pi[1:m-1]},X_{1:j-1}^{(\pi[m])}) - \frac{1}{S} \sum_{j=1}^S H(X_{j}^{(\pi[m])}| Q_{j}^{(\pi[m])}, W_{1:M},X_{1:j-1}^{(\pi[m])})
    \\&= \frac{1}{S} \sum_{j=1}^S I(W_{\pi[m:M]}; X_{j}^{(\pi[m])}| Q_{1:S}^{(\pi[m])}, W_{\pi[1:m-1]},X_{1:j-1}^{(\pi[m])})
    \\&= \frac{1}{S} L_w + \frac{1}{S} I(W_{\pi[m:M]} ; Q_{1:S}^{(\pi[m])}, X_{1:S}^{(\pi[m])}| W_{\pi[1:m]})
    \\&=  \frac{1}{S} L_w + \frac{1}{S} I(W_{\pi[m+1:M]} ; Q_{1:S}^{(\pi[m])}, X_{1:S}^{(\pi[m])}| W_{\pi[1:m]}).
  \end{align}
\end{proof}
By iteratively applying Lemma 3 to Lemma 2, the capacity is obtained.
The necessary and sufficient conditions for achieving the capacity are that the inequalities in the derivations of Lemma 2 and Lemma 3 are satisfied with equal signs.
Thus, the following equations are the necessary and sufficient conditions for achieving the capacity.
\begin{align}
  &H(X_{1:S}^{(\pi[1])}|Q_{1:S}^{(\pi[1])}) = \sum_{j=1}^S H(X_j^{(\pi[1])}|Q_{1:S}^{(\pi[1])})\label{enum:m1},
  \\&I(W_{\pi[m:M]} ; Q_{1:S}^{(\pi[m-1])}, X_{1:S}^{(\pi[m-1])}| W_{\pi[1:m-1]}) = \frac{1}{S} \sum_{j=1}^S I(W_{\pi[m:M]} ; Q_{1:S}^{(\pi[m-1])}, X_{j}^{(\pi[m-1])}| W_{\pi[1:m-1]})\label{enum:m2},
  \\&\frac{1}{S} \sum_{j=1}^S H(X_{j}^{(\pi[m])}| Q_{1:S}^{(\pi[m])}, W_{\pi[1:m-1]}) = \frac{1}{S} \sum_{j=1}^S H(X_{j}^{(\pi[m])}| Q_{1:S}^{(\pi[m])}, W_{\pi[1:m-1]},X_{1:j-1}^{(\pi[m])})\label{enum:m3}.
\end{align}
  Let $\bar{\mathcal{I}}$ be an arbitrary subset of $[1:M]\setminus\{m\}$, and let $\mathcal{I}$ be the complement of $\bar{\mathcal{I}}$.
  The following equations are the summary of the above equations:
  \begin{align}
    &\forall i,j \in [1:S], \forall m \in [1:M], H(X_{j}^{(m)}|X_i^{(m)},W_{m}, Q_{1:S}^{(m)}) = 0\label{enum:m5},
      \\&\forall j \in [1:S], \forall m \in [1:M], \forall \bar{\mathcal{I}}\subset [1:M]\setminus\{m\}, H(X_{j}^{(m)}| Q_{1:S}^{(m)}, W_{\bar{\mathcal{I}}},X_{i}^{(m)}) = H(X_{j}^{(m)}| Q_{1:S}^{(m)}, W_{\bar{\mathcal{I}}})\label{enum:m7}.
      \end{align}
\begin{proof}
  First, because of the following equations, (\ref{enum:m3}) is a sufficient condition for (\ref{enum:m1}).
      \begin{align}
        &\frac{1}{S} \sum_{j=1}^S H(X_{j}^{(\pi[m])}| Q_{1:S}^{(\pi[m])}, W_{\pi[1:m-1]}) = \frac{1}{S} \sum_{j=1}^S H(X_{j}^{(\pi[m])}| Q_{1:S}^{(\pi[m])}, W_{\pi[1:m-1]},X_{1:j-1}^{(\pi[m])}) \\
        &\Leftrightarrow \sum_{j=1}^S H(X_{j}^{(\pi[m])}| Q_{1:S}^{(\pi[m])}, W_{\pi[1:m-1]}) = \sum_{j=1}^S H(X_{j}^{(\pi[m])}| Q_{1:S}^{(\pi[m])}, W_{\pi[1:m-1]},X_{1:j-1}^{(\pi[m])})\\
        &\Rightarrow \sum_{j=1}^S H(X_{j}^{(\pi[1])}| Q_{1:S}^{(\pi[1])}) = \sum_{j=1}^S H(X_{j}^{(\pi[1])}| Q_{1:S}^{(\pi[1])},X_{1:j-1}^{(\pi[1])}) \label{c2}\\
        &\Leftrightarrow \sum_{j=1}^S H(X_{j}^{(\pi[1])}| Q_{1:S}^{(\pi[1])}) = H(X_{1:S}^{(\pi[1])}|Q_{1:S}^{(\pi[1])}). \label{c3}
      \end{align}
      (\ref{c2}) is because we can set $m=1$, and in which case $\pi[1:m-1]=\phi$.
      (\ref{c3}) is because of the chain rule.
      Therefore, (\ref{enum:m1}) is unnecessary.
      Next, the following equations hold for (\ref{enum:m2}):
      \begin{align}
        &I(W_{\pi[m:M]} ; Q_{1:S}^{(\pi[m-1])}, X_{1:S}^{(\pi[m-1])}| W_{\pi[1:m-1]}) = \frac{1}{S} \sum_{j=1}^S I(W_{\pi[m:M]} ; Q_{1:S}^{(\pi[m-1])}, X_{j}^{(\pi[m-1])}| W_{\pi[1:m-1]})
        \\&\Leftrightarrow I(W_{\pi[m:M]} ; X_{1:S}^{(\pi[m-1])}| Q_{1:S}^{(\pi[m-1])}, W_{\pi[1:m-1]}) = \frac{1}{S} \sum_{j=1}^S I(W_{\pi[m:M]} ; X_{j}^{(\pi[m-1])} | Q_{1:S}^{(\pi[m-1])}, W_{\pi[1:m-1]})
        \\&\Leftrightarrow H(X_{1:S}^{(\pi[m-1])} | W_{\pi[1:m-1]}, Q_{1:S}^{(\pi[m-1])})-H(X_{1:S}^{(\pi[m-1])} | W_{1:M}, Q_{1:S}^{(\pi[m-1])}) \notag
        \\&=\frac{1}{S} \sum_{j=1}^S \left(H(X_{j}^{(\pi[m-1])} | W_{\pi[1:m-1]}, Q_{1:S}^{(\pi[m-1])}) - H(X_{j}^{(\pi[m-1])} |  W_{1:M}, Q_{1:S}^{(\pi[m-1])})\right)
        \\&\Leftrightarrow H(X_{1:S}^{(\pi[m-1])} | W_{\pi[1:m-1]}, Q_{1:S}^{(\pi[m-1])}) = \frac{1}{S} \sum_{j=1}^S H(X_{j}^{(\pi[m-1])} | W_{\pi[1:m-1]},Q_{1:S}^{(\pi[m-1])}) 
        \\&\Leftrightarrow \sum_{j=1}^S H(X_{j}^{(\pi[m-1])} | W_{\pi[1:m-1]}, Q_{1:S}^{(\pi[m-1])}, X_{1:j-1}) = \frac{1}{S} \sum_{j=1}^S H(X_{j}^{(\pi[m-1])} | W_{\pi[1:m-1]}, Q_{1:S}^{(\pi[m-1])}) \label{c1}
        \\&\Leftrightarrow 
        \forall i,j\in [1:S] ,H(X_j^{(\pi[m-1])} | W_{\pi[1:m-1]},Q_{1:S}^{(\pi[m-1])},X_i^{(\pi[m-1])}) = 0 
        \\&\Leftrightarrow H(X_j^{(\pi[m])} | W_{\pi[1:m]},Q_{1:S}^{(\pi[m])},X_i^{(\pi[m])}) = 0 
        \\&\Leftrightarrow H(X_j^{(m)} | W_{m},Q_{1:S}^{(m)},X_i^{(m)}) = 0 .\label{m0}
    \end{align}  
    (\ref{c1}) is because of the chain rule. (\ref{m0}) is because we can set $\pi[1]=m$.
    And for (\ref{enum:m3}), the following equations hold:
    \begin{align}
      &\sum_{j=1}^S H(X_{j}^{(\pi[m])}| Q_{1:S}^{(\pi[m])}, W_{\pi[1:m-1]}) = \sum_{j=1}^S H(X_{j}^{(\pi[m])}| Q_{1:S}^{(\pi[m])}, W_{\pi[1:m-1]},X_{1:j-1}^{(\pi[m])}) \\
      &\Leftrightarrow H(X_{i}^{(\pi[m])}|Q_{1:S}^{(\pi[m])},W_{\pi[1:m-1]},X_j^{(\pi[m])})=H(X_i^{(\pi[m])}|Q_{1:S}^{(\pi[m])},W_{\pi[1:m-1]}) \notag
      \\& \Leftrightarrow \forall  \bar{\mathcal{I}}\subset [1:M]\setminus\{m\}, H(X_{j}^{(m)}| Q_{1:S}^{(m)}, W_{\bar{\mathcal{I}}},X_{i}^{(m)}) = H(X_{j}^{(m)}| Q_{1:S}^{(m)}, W_{\bar{\mathcal{I}}}) \label{c4}.
    \end{align}
  \end{proof}
 (\ref{c4}) is because we can set $\pi[m]=m$.
  Note that the query $Q_j^{(m)}$ to server $j$ is written as $Q_j^{(m)}=(Q_j^{(m)}[1],\dots,Q_j^{(m)}[M])$, 
  let $\bar{m} \coloneq [1:M]\setminus \{m\}$, 
  the necessary and sufficient conditions for capacity can be rewritten as follows:
  \begin{align}
  &\forall j \in [1:S], \forall m \in [1:M], \mathrm{rank}Q_{1:S}^{(m)}[\bar{m}]=\mathrm{rank}\ Q_j^{(m)}[\bar{m}] \label{enum:m8},  
  \\&\forall j \in [1:S], \forall m \in [1:M], \forall \bar{\mathcal{I}}\subset [1:M]\setminus\{m\}, \mathrm{rank}\ Q_{1:S}^{(m)}[\mathcal{I}]^{\top}=\sum_{j=1}^{S} \mathrm{rank}\ Q_j^{(m)}[\mathcal{I}] \label{enum:m10}. 
  \end{align}
\begin{proof}
  For all $i,j\in [1:S]$, the server responses are given as follows:
  \begin{align}
   X_i^{(m)}=Q_i^{(m)} [\bar{m}]W_{\bar{m}} + Q_i^{(m)}[m]W_m, \\
   X_j^{(m)}=Q_j^{(m)} [\bar{m}]W_{\bar{m}} + Q_j^{(m)}[m]W_m.
  \end{align}
  (\ref{enum:m5}) indicates $X_j^{(m)}$ is determined by $X_i^{(m)}$ 
  without $W_{\bar{m}}$, so $Q_j^{(m)}[\bar{m}]$ is a multiple of $Q_i^{(m)}[\bar{m}]$.  
  Therefore, (\ref{enum:m8}) holds.
  The responses from the servers can also be given as follows:
  \begin{align}
    X_i^{(m)}=Q_i^{(m)} [\bar{\mathcal{I}}]W_{\bar{\mathcal{I}}} + Q_i^{(m)}[\mathcal{I}]W_\mathcal{I}, \\
    X_j^{(m)}=Q_j^{(m)} [\bar{\mathcal{I}}]W_{\bar{\mathcal{I}}} + Q_j^{(m)}[\mathcal{I}]W_\mathcal{I}.
   \end{align}
   (\ref{enum:m7}) indicates $X_j^{(m)}$ is not obtained by $X_i^{(m)}$ without knowing $W_\mathcal{I}$,
   so $Q_i^{(m)}[\mathcal{I}]$ and $Q_j^{(m)}[\mathcal{I}]$ are independent.
   Therefore, (\ref{enum:m10}) holds.
\end{proof}

\section{Example}
\subsection{PIR method Description}
The user generates the queries as follows:
\begin{align}
  Q_1^{(m)}=
  &(
    Z_1 , \dots , Z_1^{S-1} ,
    \dots ,
    Z_{m-1} , \dots , Z_{m-1}^{S-1} ,
    z_1 , \dots , z_1^{S-1} ,
    Z_{m+1} , \dots , Z_{m+1}^{S-1} ,
    \dots ,
    Z_M , \dots , Z_M^{S-1}
  ), \notag  \\
  Q_2^{(m)}=
  &(
    Z_1 , \dots , Z_1^{S-1} ,
    \dots , 
    Z_{m-1} , \dots , Z_{m-1}^{S-1} ,
    z_2 , \dots , z_2^{S-1} ,
    Z_{m+1} , \dots , Z_{m+1}^{S-1} ,
    \dots ,
    Z_M , \dots , Z_M^{S-1}
  ), \notag \\
  &\vdots \notag \\
  Q_S^{(m)}=
  &(
    Z_1 , \dots , Z_1^{S-1} ,
    \dots ,
    Z_{m-1} , \dots , Z_{m-1}^{S-1} ,
    z_S , \dots , z_S^{S-1} ,
    Z_{m+1} , \dots , Z_{m+1}^{S-1} ,
    \dots ,
    Z_M , \dots , Z_M^{S-1}
  ). 
\end{align}
Assume that $Z_1,\dots, Z_{m-1}, Z_{m+1},\dots, Z_M$ follow a uniform distribution on the finite field $\mathbb{F}_S$ 
and $z_j$ is also uniformly chosen to be different from each other on $\mathbb{F}_S$.
Each server divides the messages into $S-1$ equal-length parts as follows:
\begin{equation}
  k\in [1:M],
  W_k=
(w_{k,1} , \dots , w_{k,S-1}).
\end{equation}
Let $w_{k,1},\dots,w_{k,S-1}$ follow a uniform distribution on $\mathbb{F}_S$.
The responses and the queries are constructed as follows:
\begin{equation}
  \begin{split}
    \begin{pmatrix}
      X_{1}\\
      X_{2}\\ 
      \vdots \\
      X_{S}    
    \end{pmatrix}
  = 
  \begin{pmatrix}
    1 & z_1 & z_1^2 & \dots & z_1^{S-1}\\
    1 & z_2 & z_2^2 & \dots & z_2^{S-1}\\
    \vdots & \vdots & \vdots & \ddots & \vdots \\
    1 & z_S & z_S^2 & \dots & z_S^{S-1}\\    
  \end{pmatrix} 
  \begin{pmatrix}
    \beta\\
    w_{m,1}\\
    \vdots\\
    w_{m,S-1}
  \end{pmatrix},
  \beta \coloneq \sum_{k\in[1:M]\setminus \{m\}}
  \begin{pmatrix}
    Z_k & Z_k^2 & \dots & Z_k^{S-1}
  \end{pmatrix} 
  \begin{pmatrix}
    w_{k,1}\\\vdots\\w_{k,S-1}   
  \end{pmatrix}.
\end{split}
\end{equation}

\subsection{Correctness}
In this method, the following matrix can be a decoding matrix:
\begin{align}
D_m=
\begin{pmatrix}
  0 & 1 & 0 & \dots & 0 \\
  0 & 0 & 1 & \dots & 0 \\
  \vdots & \vdots & 0 & \ddots & \vdots\\
  0 & 0 & 0 & \dots & 1   
\end{pmatrix}
\begin{pmatrix}
  1 & z_1 & z_1^2 & \dots & z_1^{S-1} \\ 
  1 & z_2 & z_2^2 & \dots & z_2^{S-1} \\
  \vdots & \vdots & \vdots & \ddots & \vdots \\
  1 & z_S & z_S^2 & \dots & z_S^{S-1}
\end{pmatrix}^{-1}.
\end{align}
\begin{proof}
  {\footnotesize
  \begin{align}
    &\begin{pmatrix}
      X_{1}\\
      X_{2}\\ 
      \vdots \\
      X_{S}    
    \end{pmatrix}
    = \notag \\
    &\begin{pmatrix}
      Z_1 & \dots & Z_1^{S-1} &
      \dots & 
      Z_{m-1} & \dots & Z_{m-1}^{S-1} &
      z_1 & \dots & z_1^{S-1} &
      Z_{m+1} & \dots & Z_{m+1}^{S-1} &
      \dots & 
      Z_M & \dots & Z_M^{S-1} \\
      Z_1 & \dots & Z_1^{S-1} &
      \dots & 
      Z_{m-1} & \dots & Z_{m-1}^{S-1} &
      z_2 & \dots & z_2^{S-1} &
      Z_{m+1} & \dots & Z_{m+1}^{S-1} &
      \dots & 
      Z_M & \dots & Z_M^{S-1} \\
      \vdots & \vdots & \vdots &
      \vdots &
      \vdots & \vdots & \vdots &
      \vdots & \vdots & \vdots &
      \vdots & \vdots & \vdots &
      \vdots &
      \vdots & \vdots & \vdots \\      
      Z_1 & \dots & Z_1^{S-1} &
      \dots & 
      Z_{m-1} & \dots & Z_{m-1}^{S-1} &
      z_S & \dots & z_S^{S-1} &
      Z_{m+1} & \dots & Z_{m+1}^{S-1} &
      \dots & 
      Z_M & \dots & Z_M^{S-1}
    \end{pmatrix}
    \begin{pmatrix}
      w_{1,1}\\\vdots\\w_{1,S-1}\\ \vdots \\ w_{m,1}\\\vdots\\w_{m,S-1}\\\vdots \\w_{M,1}\\\vdots\\w_{M,S-1}
    \end{pmatrix} \notag \\
    & = \begin{pmatrix}
      Z_1 & \dots & Z_1^{S-1} &
      \dots & 
      Z_{m-1} & \dots & Z_{m-1}^{S-1} &
      Z_{m+1} & \dots & Z_{m+1}^{S-1} &
      \dots & 
      Z_M & \dots & Z_M^{S-1} &
      z_1 & \dots & z_1^{S-1} \\
      Z_1 & \dots & Z_1^{S-1} &
      \dots & 
      Z_{m-1} & \dots & Z_{m-1}^{S-1} &
      Z_{m+1} & \dots & Z_{m+1}^{S-1} &
      \dots & 
      Z_M & \dots & Z_M^{S-1} &
      z_2 & \dots & z_2^{S-1} \\
      \vdots & \vdots & \vdots &
      \vdots &
      \vdots & \vdots & \vdots &
      \vdots & \vdots & \vdots &
      \vdots & \vdots & \vdots &
      \vdots &
      \vdots & \vdots & \vdots \\      
      Z_1 & \dots & Z_1^{S-1} &
      \dots & 
      Z_{m-1} & \dots & Z_{m-1}^{S-1} &
      Z_{m+1} & \dots & Z_{m+1}^{S-1} &
      \dots & 
      Z_M & \dots & Z_M^{S-1}&
      z_S & \dots & z_S^{S-1} 
    \end{pmatrix}
    \begin{pmatrix}
      w_{1,1}\\\vdots\\w_{1,S-1}\\ \vdots \\ w_{M,1}\\\vdots\\w_{M,S-1} \\ w_{m,1}\\\vdots\\w_{m,S-1}
    \end{pmatrix}.
  \end{align}
  }
Let 
\begin{align}
  &A \coloneq
  (
    Z_1 , \dots , Z_1^{S-1} ,
    \dots ,
    Z_{m-1} , \dots , Z_{m-1}^{S-1} ,
    Z_{m+1} , \dots , Z_{m+1}^{S-1} ,
    \dots ,
    Z_M , \dots , Z_M^{S-1} 
  ),\\
  &B \coloneq
  \begin{pmatrix}
    w_{1,1}\\\vdots\\w_{1,S-1}\\ \vdots \\ w_{M,1}\\\vdots\\w_{M,S-1}
  \end{pmatrix} ,
\end{align}  
then the responses can be rewritten as follows:
\begin{align}
  \begin{pmatrix}
    X_{1}\\
    X_{2}\\ 
    \vdots \\
    X_{S}    
  \end{pmatrix}
  =
  \begin{pmatrix}
    A & z_1 & \dots & z_1^{S-1} \\
    A & z_2 & \dots & z_2^{S-1} \\
    \vdots & \vdots & \vdots & \vdots \\      
    A & z_S & \dots & z_S^{S-1} 
  \end{pmatrix}
  \begin{pmatrix}
    B \\w_{m,1}\\\vdots\\w_{m,S-1}
  \end{pmatrix}. 
\end{align}
Let $\beta \coloneq AB$, then the following equation holds:
\begin{align}
  \begin{pmatrix}
    X_{1}\\
    X_{2}\\ 
    \vdots \\
    X_{S}    
  \end{pmatrix}
  =
  \begin{pmatrix}
    1 & z_1 & \dots & z_1^{S-1} \\
    1 & z_2 & \dots & z_2^{S-1} \\
    \vdots & \vdots & \vdots & \vdots \\      
    1 & z_S & \dots & z_S^{S-1} 
  \end{pmatrix}
  \begin{pmatrix}
    \beta \\w_{m,1}\\\vdots\\w_{m,S-1}
  \end{pmatrix} .
\end{align}
Therefore, from the regularity of Vandermonde square matrices, 
the following equations hold:
\begin{align}
  \begin{pmatrix}
    1 & z_1 & \dots & z_1^{S-1} \\
    1 & z_2 & \dots & z_2^{S-1} \\
    \vdots & \vdots & \vdots & \vdots \\      
    1 & z_S & \dots & z_S^{S-1} 
  \end{pmatrix}^{-1}
  \begin{pmatrix}
    X_{1}\\
    X_{2}\\ 
    \vdots \\
    X_{S}    
  \end{pmatrix}
  =
  \begin{pmatrix}
    \beta \\w_{m,1}\\\vdots\\w_{m,S-1}
  \end{pmatrix} .\label{inv}
\end{align}
To further exclude $\beta$, (\ref{inv}) is transformed as follows:
\begin{align}
  \begin{pmatrix}
    0 & 1 & 0 & \dots & 0 \\
    0 & 0 & 1 & \dots & 0 \\
    \vdots & \vdots & 0 & \ddots & \vdots\\
    0 & 0 & 0 & \dots & 1   
  \end{pmatrix}
  \begin{pmatrix}
    1 & z_1 & z_1^2 & \dots & z_1^{S-1} \\ 
    1 & z_2 & z_2^2 & \dots & z_2^{S-1} \\
    \vdots & \vdots & \vdots & \ddots & \vdots \\
    1 & z_S & z_S^2 & \dots & z_S^{S-1}
  \end{pmatrix}^{-1}
  \begin{pmatrix}
    X_{1}\\
    X_{2}\\ 
    \vdots \\
    X_{S}    
  \end{pmatrix}
  =
  \begin{pmatrix}
    w_{m,1}\\\vdots\\w_{m,S-1}
  \end{pmatrix} .
\end{align}
\end{proof}

\subsection{Privacy}
Assume $m=1$, and for $j\in [2:S]$, let
\begin{align}
  Q_j^{(m)}=(z_j, \dots, z_j^{S-1}, Z_2, \dots, Z_2^{S-1}, \dots, Z_M, \dots, Z_M^{S-1}).
\end{align}
Note that $z_j \in \mathbb{F}_S \setminus \{z_1\}$. So $(Q_2^{(m)},\dots,Q_S^{(m)})$ has $(S-1)!$ possible realizations.
Similarly, for all $m \in [1:M]$, there are $(S-1)!$ possible realizations each, for a total of $M(S-1)!$.
Therefore, the following equations hold.
\begin{align}
   &|\{Q_{1:S}^{(m)} | Q_{j}^{(m)} = q_j\}|= (S-1)!,
\\ &|\{Q_{1:S}^{(\theta)} | Q_{j}^{(\theta)} = q_j\}|= M(S-1)!.
\end{align}
So the following privacy condition is satisfied. 
\begin{align}
  |\{Q_{1:S}^{(\theta)} | Q_{j}^{(\theta)} = q_j\}|=M|\{Q_{1:S}^{(m)} | Q_{j}^{(m)} = q_j\}|.
\end{align}

\subsection{Capacity}
First, for (\ref{enum:m8}), the queries without the submatrix corresponding to index $m$ are as follows:
\begin{align}
  Q_1^{(m)}=
  &(
    Z_1 , \dots , Z_1^{S-1} ,
    \dots ,
    Z_{m-1} , \dots , Z_{m-1}^{S-1} ,
    Z_{m+1} , \dots , Z_{m+1}^{S-1} ,
    \dots ,
    Z_M , \dots , Z_M^{S-1}
  ), \notag  \\
  Q_2^{(m)}=
  &(
    Z_1 , \dots , Z_1^{S-1} ,
    \dots , 
    Z_{m-1} , \dots , Z_{m-1}^{S-1} ,
    Z_{m+1} , \dots , Z_{m+1}^{S-1} ,
    \dots ,
    Z_M , \dots , Z_M^{S-1}
  ), \notag \\
    &\vdots \notag \\
  Q_S^{(m)}=
  &(
    Z_1 , \dots , Z_1^{S-1} ,
    \dots ,
    Z_{m-1} , \dots , Z_{m-1}^{S-1} ,
    Z_{m+1} , \dots , Z_{m+1}^{S-1} ,
    \dots ,
    Z_M , \dots , Z_M^{S-1}
  ). 
\end{align}
Therefore, the rank of the queries is as follows:
\begin{align}
  \mathrm{rank}
\begin{pmatrix}
  Q_1^{(m)}[\bar{m}]\\
  Q_2^{(m)}[\bar{m}]\\
  \vdots\\
  Q_S^{(m)}[\bar{m}]
\end{pmatrix}
=
\begin{cases}
  0 \quad (\mathrm{if\ one\ of\ } Q_j^{(m)}[\bar{m}] \mathrm{\ is\ a\ zero\ vector})\\ 
  1 \quad (\mathrm{otherwise})
\end{cases}.  
\end{align}
Note that
\begin{align}
  \mathrm{rank}\ Q_1^{(m)}[\bar{m}]=\dots=  \mathrm{rank}\ Q_S^{(m)}[\bar{m}]=
  \begin{cases}
    0 \quad (\mathrm{if\ one\ of\ } Q_j^{(m)}[\bar{m}] \mathrm{\ is\ a\ zero\ vector})\\ 
    1 \quad (\mathrm{otherwise})
  \end{cases} .
\end{align}
Thus, (\ref{enum:m8}) is satisfied. Next, for all $\mathcal{I}$ which is the complement of $\bar{\mathcal{I}} \subset [1:M]\setminus \{m\}$, the queries are as follows:
\begin{align}
  Q_1^{(m)}[\mathcal{I}]=
  &(
    Z_{i_1} , \dots , Z_{i_1}^{S-1} ,
    \dots , 
    z_1 , \dots , z_1^{S-1} ,
    \dots , 
    Z_{i_n} , \dots , Z_{i_n}^{S-1}
  ), \notag\\
  Q_2^{(m)}[\mathcal{I}]=
  &(
    Z_{i_1} , \dots , Z_{i_1}^{S-1} ,
    \dots , 
    z_2 , \dots , z_2^{S-1} ,
    \dots , 
    Z_{i_n} , \dots , Z_{i_n}^{S-1}
  ), \notag \\
  &\vdots \notag \\
  Q_S^{(m)}[\mathcal{I}]=
  &(
    Z_{i_1} , \dots , Z_{i_1}^{S-1} ,
    \dots ,
    z_S , \dots , z_S^{S-1} ,
    \dots , 
    Z_{i_n} , \dots , Z_{i_n}^{S-1}
  ).
\end{align}
So clearly, the following equations hold:
\begin{align}
  \mathrm{rank}
\begin{pmatrix}
  Q_1^{(m)}[\mathcal{I}]\\
  Q_2^{(m)}[\mathcal{I}]\\
  \vdots\\
  Q_S^{(m)}[\mathcal{I}]
\end{pmatrix}
=
\begin{cases}
S-1 \quad (\mathrm{if\ one\ of\ }Q_j^{(m)}[\mathcal{I}]\mathrm{\ is\ a\ zero\ vector})\\ 
S \quad (\mathrm{otherwise})
\end{cases}
\end{align}
Therefore, (\ref{enum:m10}) is satisfied. 

\section{Colluding PIR}
\subsection{Conditions for Correctness}
The correctness condition is the same as for the standard PIR: there exists a matrix $D_m$ and the following equation is satisfied:
\begin{align}
    D_m \begin{pmatrix}Q_1^{(m)}\\\vdots \\ Q_S^{(m)}\end{pmatrix} = 
    \begin{pmatrix}
      \mathbf{O}_{(m-1)L_w \times L_w}\\
      \mathbf{E}\\
      \mathbf{O}_{(M-m)L_w \times L_w}
     \end{pmatrix}^{\top}.
  \end{align}
\subsection{Conditions for Privacy}
\begin{theorem}[Conditions for Privacy of colluding PIR]
    Let $\mathcal{T}$ be the index set of $T$ colluding servers, 
    the conditions for privacy are given as follows:
    \begin{align}
      &\mathrm{Pr}\{Q_{\mathcal{T}}^{(\theta)}=q | \theta=m\} = \mathrm{Pr}\{Q_{\mathcal{T}}^{(\theta)}=q | \theta=m^{\prime}\}
      \\& \Leftrightarrow |\{Q_{1:S}^{(\theta)} | Q_{\mathcal{T}}^{(\theta)} = q_{\mathcal{T}}\}|=M|\{Q_{1:S}^{(m)} | Q_{\mathcal{T}}^{(m)} = q_{\mathcal{T}}\}|. \notag
    \end{align}
    \end{theorem}
    \begin{proof}
        \begin{align}
            &\mathrm{Pr}\{Q_{\mathcal{T}}^{(\theta)}=q | \theta=m\} = \mathrm{Pr}\{Q_{\mathcal{T}}^{(\theta)}=q | \theta=m^{\prime}\}
            \\ &\Leftrightarrow \mathrm{Pr}\{Q_{\mathcal{T}}^{(\theta)}=q_{\mathcal{T}}, \theta=m\} = \mathrm{Pr}\{Q_{\mathcal{T}}^{(\theta)}=q_{\mathcal{T}} | \theta=m^{\prime}\} \mathrm{Pr}\{\theta = m\}
            \\ &\Leftrightarrow \mathrm{Pr}\{Q_{\mathcal{T}}^{(\theta)}=q_{\mathcal{T}}, \theta=m\} = \mathrm{Pr}\{Q_{\mathcal{T}}^{(\theta)}=q_{\mathcal{T}} , \theta=m^{\prime}\} \mathrm{Pr}\{\theta = m\} / \mathrm{Pr}\{\theta=m^{\prime}\}
            \\ &\Leftrightarrow \mathrm{Pr}\{\theta=m | Q_{\mathcal{T}}^{(\theta)}=q_{\mathcal{T}}\} = \mathrm{Pr}\{\theta=m^{\prime} | Q_{\mathcal{T}}^{(\theta)}=q_{\mathcal{T}}\} \mathrm{Pr}\{\theta = m\} / \mathrm{Pr}\{\theta=m^{\prime}\}
            \\ &\Leftrightarrow \mathrm{Pr}\{\theta=m | Q_{\mathcal{T}}^{(\theta)}=q_{\mathcal{T}}\} = \mathrm{Pr}\{\theta=m^{\prime} | Q_{\mathcal{T}}^{(\theta)}=q_{\mathcal{T}}\}
            \\ &\Leftrightarrow \mathrm{Pr}\{\theta=m | Q_{\mathcal{T}}^{(\theta)}=q_{\mathcal{T}}\} = \frac{1}{M} 
            \\ &\Leftrightarrow H(\theta | Q_{\mathcal{T}}^{(\theta)}=q_{\mathcal{T}}) = \log M 
            \\ &\Leftrightarrow H(\theta | Q_{\mathcal{T}}^{(\theta)}=q_{\mathcal{T}}) - H(\theta | Q_{\bar{\mathcal{T}}}^{(\theta)},Q_{\mathcal{T}}^{(\theta)}=q_{\mathcal{T}})= \log M  
            \\ &\Leftrightarrow I(\theta;Q_{\bar{\mathcal{T}}}^{(\theta)} | Q_{\mathcal{T}}^{(\theta)}=q_{\mathcal{T}}) = \log M
            \\ &\Leftrightarrow H(Q_{\bar{\mathcal{T}}}^{(\theta)} | Q_{\mathcal{T}}^{(\theta)}=q_{\mathcal{T}}) - H(Q_{\bar{\mathcal{T}}}^{(\theta)} | \theta ,Q_{\mathcal{T}}^{(\theta)}=q_{\mathcal{T}}) = \log M
            \\ &\Leftrightarrow H(Q_{\bar{\mathcal{T}}}^{(\theta)} | Q_{\mathcal{T}}^{(\theta)}=q_{\mathcal{T}}) = \log M + H(Q_{\bar{\mathcal{T}}}^{(\theta)} | \theta ,Q_{\mathcal{T}}^{(\theta)}=q_{\mathcal{T}}) 
            \\ &\Leftrightarrow H(Q_{\bar{\mathcal{T}}}^{(\theta)} | Q_{\mathcal{T}}^{(\theta)}=q_{\mathcal{T}}) = \log M + \sum_{i=1}^M P(\theta=i) H(Q_{\bar{\mathcal{T}}}^{(\theta)} | \theta=i ,Q_{\mathcal{T}}^{(\theta)}=q_{\mathcal{T}}) 
            \\ &\Leftrightarrow H(Q_{\bar{\mathcal{T}}}^{(\theta)} | Q_{\mathcal{T}}^{(\theta)}=q_{\mathcal{T}}) = \log M + \sum_{i=1}^M \frac{1}{M}H(Q_{\bar{\mathcal{T}}}^{(\theta)} | \theta=i ,Q_{\mathcal{T}}^{(\theta)}=q_{\mathcal{T}}) 
            \\ &\Leftrightarrow H(Q_{\bar{\mathcal{T}}}^{(\theta)} | Q_{\mathcal{T}}^{(\theta)}=q_{\mathcal{T}}) = \log M + H(Q_{\bar{\mathcal{T}}}^{(\theta)} | \theta=m ,Q_{\mathcal{T}}^{(\theta)}=q_{\mathcal{T}}) 
            \\ &\Leftrightarrow H(Q_{\bar{\mathcal{T}}}^{(\theta)} | Q_{\mathcal{T}}^{(\theta)}=q_{\mathcal{T}}) = \log M + H(Q_{\bar{\mathcal{T}}}^{(m)} | Q_{\mathcal{T}}^{(m)}=q_{\mathcal{T}}) 
            \\ &\Leftrightarrow |\{Q_{1:S}^{(\theta)} | Q_{\mathcal{T}}^{(\theta)} = q_{\mathcal{T}}\}|=M|\{Q_{1:S}^{(m)} | Q_{\mathcal{T}}^{(m)} = q_{\mathcal{T}}\}|.
           \end{align}
\end{proof}

\subsection{Conditions for Capacity Achievability}
Assume that there are $T$ colluding servers, the capacity is given as follows \cite{8119895}:
\begin{theorem}[Conditions for Capacity of colluding PIR]  
     \begin{align}
        \mathcal{R} \leq \left(1+\frac{T}{S}+\frac{T^2}{S^2}+\frac{T^3}{S^3}+\dots+\frac{T^{M-1}}{S^{M-1}}\right)^{-1}=\frac{1-\frac{T}{S}}{1-(\frac{T}{S})^M}.
     \end{align}
    \end{theorem}
   \begin{proof}
    Let $\mathcal{Q}$ be the query realization set $\{Q_{1:S}^{(m)} | m \in [1:M]\}$,
    \begin{align}
        \mathcal{T} \subset [1:S], |\mathcal{T}|=T, 
        \mathcal{H}_T \coloneq \frac{1}{\begin{pmatrix}S\\T\end{pmatrix}} \sum_{\mathcal{T}} \frac{H(X_{\mathcal{T}}^{(\pi[1])}|\mathcal{Q})}{T}.
    \end{align}
    Then, the following equations hold.\\
    \begin{align}
    &M L_w=H(W_1,\dots,W_M|\mathcal{Q})\\
    &= H(X_{\mathcal{T}}^{(\pi[1])},X_{\bar{\mathcal{T}}}^{(\pi[1])},\dots,X_{\bar{\mathcal{T}}}^{(\pi[M])}|\mathcal{Q}) \\ 
    &= H(X_{\mathcal{T}}^{(\pi[1])},X_{\bar{\mathcal{T}}}^{(\pi[1])}| \mathcal{Q}) + H(X_{\bar{\mathcal{T}}}^{(\pi[2])},\dots,X_{\bar{\mathcal{T}}}^{(\pi[M])}|X_{\mathcal{T}}^{(\pi[1])},X_{\bar{\mathcal{T}}}^{(\pi[1])},\mathcal{Q})\\
    &\leq S\mathcal{H}_T + H(X_{\bar{\mathcal{T}}}^{(\pi[2])},\dots,X_{\bar{\mathcal{T}}}^{(\pi[M])}|X_{\mathcal{T}}^{(\pi[1])},X_{\bar{\mathcal{T}}}^{(\pi[1])},\mathcal{Q}) \label{han}\\
    &\leq S\mathcal{H}_T + H(X_{\bar{\mathcal{T}}}^{(\pi[2])},\dots,X_{\bar{\mathcal{T}}}^{(\pi[M])}|X_{\mathcal{T}}^{(\pi[1])},W_{\pi[1]},\mathcal{Q}) \label{con0}\\
    &= S\mathcal{H}_T + H(X_{\bar{\mathcal{T}}}^{(\pi[2])}|X_{\mathcal{T}}^{(\pi[1])},W_{\pi[1]},\mathcal{Q}) + H(X_{\bar{\mathcal{T}}}^{(\pi[3])},\dots,X_{\bar{\mathcal{T}}}^{(\pi[M])}|X_{\mathcal{T}}^{(\pi[1])},X_{\mathcal{T}}^{(\pi[2])},W_{\pi[1]},\mathcal{Q}) \label{ind1}\\
    &\leq S\mathcal{H}_T + \sum_{n \in \bar{\mathcal{T}}} H(X_{n}^{(\pi[2])}|X_{\mathcal{T}}^{(\pi[1])},W_{\pi[1]},\mathcal{Q})+ H(X_{\mathcal{T}}^{(\pi[1])},X_{\bar{\mathcal{T}}}^{(\pi[3])},\dots,X_{\bar{\mathcal{T}}}^{(\pi[M])}|W_{\pi[1]},W_{\pi[2]},\mathcal{Q})  \label{ind2} \\ \notag
    & \quad -H(X_{\mathcal{T}}^{(\pi[1])}|W_{\pi[1]},W_{\pi[2]},\mathcal{Q})\\
    &= S\mathcal{H}_T + \sum_{n \in \bar{\mathcal{T}}} H(X_{n}^{(\pi[2])}|X_{\mathcal{T}}^{(\pi[1])},W_{\pi[1]},\mathcal{Q})+ (M-2) L_w -H(X_{\mathcal{T}}^{(\pi[1])}|W_{\pi[1]},W_{\pi[2]},\mathcal{Q})  \\ 
    &\Rightarrow ML_w-(M-2)L_w + \frac{1}{\begin{pmatrix}S\\T\end{pmatrix}} \sum_{\mathcal{T}} H(X_{\mathcal{T}}^{(\pi[1])}|W_{\pi[1]},W_{\pi[2]},\mathcal{Q}) \\
    &\leq S\mathcal{H}_T + \frac{1}{\begin{pmatrix}S\\T\end{pmatrix}} \sum_{\mathcal{T}} \sum_{n \notin \mathcal{T}} H(X_n^{(\pi[2])} | X_{\mathcal{T}}^{(\pi[1])},W_{\pi[1]},\mathcal{Q}) \label{ave2}\\
    &\leq S\mathcal{H}_T + \frac{1}{\begin{pmatrix}S\\T\end{pmatrix}} \sum_{\mathcal{T}} \sum_{n \notin \mathcal{T}} \sum_{\mathcal{T}^{\prime},|\mathcal{T}^{\prime}|=T-1}  \frac{H(X_n^{(\pi[2])} | X_{\mathcal{T}^{\prime}}^{(\pi[1])},W_{\pi[1]},\mathcal{Q})}{T} \label{ave1}\\
    &\leq S\mathcal{H}_T + \frac{1}{\begin{pmatrix}S\\T\end{pmatrix}} \sum_{\mathcal{T}} (S-T) \frac{H(X_{\mathcal{T}}^{(\pi[2])} | W_{\pi[1]},\mathcal{Q})}{T} \label{ind3}\\
    &\leq S\mathcal{H}_T + \left(\frac{S}{T}-1 \right) \frac{1}{\begin{pmatrix}S\\T\end{pmatrix}} \sum_{\mathcal{T}} H(X_{\mathcal{T}}^{(\pi[2])}, X_{\bar{\mathcal{T}}}^{(\pi[1])}| W_{\pi[1]},\mathcal{Q}) \label{ind4}\\
    &= S\mathcal{H}_T + \left(\frac{S}{T}-1 \right) (H(X_{1:S}^{(\pi[2])},W_{\pi[1]}|\mathcal{Q})-H(W_{\pi[1]}|\mathcal{Q}))\\
    &= S\mathcal{H}_T + \left(\frac{S}{T}-1 \right) (H(X_{1:S}^{(\pi[2])},W_{\pi[1]}|\mathcal{Q})-H(W_{\pi[1]})) \\ 
    &\leq S\mathcal{H}_T + \left(\frac{S}{T}-1 \right) (S\mathcal{H}_T -Lw) \\
    &= \frac{S^2}{T} \mathcal{H}_T - \left(\frac{S}{T}-1 \right) L_w \\
 &\Rightarrow S\mathcal{H}_T \geq L_w\left(1+\frac{T}{S} \right) + \frac{T^2}{S^2}S \frac{1}{\begin{pmatrix}S\\T\end{pmatrix}}\sum_{\mathcal{T}} \frac{H(X_{\mathcal{T}}^{(\pi[1])} | W_{\pi[1]},W_{\pi[2]},\mathcal{Q})}{T}\\
    &=L_w \left(1+\frac{T}{S} \right) + \frac{T^2}{S^2}S \mathcal{H}_{T}^{\prime}\\
 &\Rightarrow S\mathcal{H}_T \geq L_w \left(1+\frac{T}{S} + \dots + \frac{T^{M-1}}{S^{M-1}}\right).
\end{align}
\end{proof}
Note that $\mathcal{H}_{T}^{\prime} \coloneq \frac{1}{\begin{pmatrix}S\\T\end{pmatrix}}\sum_{\mathcal{T}} \frac{H(X_{\mathcal{T}}^{(\pi[1])} | W_{\pi[1]},W_{\pi[2]},\mathcal{Q})}{T}$.
The necessary and sufficient condition for achieving the capacity is that the inequality in the derivations of Theorem 3 is satisfied with equal signs.
Therefore, the following equations are the necessary and sufficient conditions for achieving the capacity:
    \begin{align}
        &\forall \mathcal{T},H(X_{1:S}^{(\pi[m])}|W_{\pi[1:m-1]},\mathcal{Q})=\frac{S}{\begin{pmatrix}S\\T\end{pmatrix}}\sum_{\mathcal{T}} \frac{H(X_{\mathcal{T}}^{(\pi[m])} | W_{\pi[1:m-1]},\mathcal{Q})}{T} \quad (\mathrm{Due\ to\ (\ref{han})}), \label{han2}\\
        &\forall \mathcal{T},H(X_{\bar{\mathcal{T}}}^{(\pi[m])}|X_{\mathcal{T}}^{(\pi[m-1])},W_{\pi[1:m-1]},\mathcal{Q})=\sum_{n \in \mathcal{\bar{T}}} H(X_n^{(\pi[m])} | X_{\mathcal{T}}^{(\pi[m-1])},W_{\pi[1:m-1]},\mathcal{Q})\quad (\mathrm{Due\ to\ (\ref{ind2})}),\\
        &\forall \mathcal{T},H(X_{\mathcal{T}}^{(\pi[m])}|W_{\pi[1:m+1]},\mathcal{Q}) = \frac{1}{\begin{pmatrix}S\\T\end{pmatrix}} \sum_{\mathcal{T}} H(X_{\mathcal{T}}^{(\pi[m+1])}|W_{\pi[1:m]},\mathcal{Q}) \quad(\mathrm{Due\ to\ (\ref{ave2})}),\\
        &\forall \mathcal{T},\forall n \notin \mathcal{T}, \forall \mathcal{T}^{\prime}, |\mathcal{T}^{\prime}|=T-1, \notag \\  & \quad H(X_n^{(\pi[m])} | X_{\mathcal{T}}^{(\pi[m-1])},W_{\pi[1:m-1]},\mathcal{Q})=H(X_n^{(\pi[m])}|X_{\mathcal{T}^{\prime}}^{(\pi[m-1])},W_{\pi[1:m-1]},\mathcal{Q}) \quad (\mathrm{Due\ to\ (\ref{ave1})}), \label{ind5}\\
        &\forall \mathcal{T},\forall n \in [1:S], \sum_{n\in \mathcal{T}} H(X_n^{(\pi[m])} | X_{\mathcal{T}\setminus \{n\}}^{(\pi[m-1])},W_{\pi[1:m-1]},\mathcal{Q}) = H(X_{\mathcal{T}}^{(\pi[m])}| W_{\pi[1:m-1]},\mathcal{Q}) \quad (\mathrm{Due\ to\ (\ref{ind3})}), \label{ind6}\\
        &\forall \mathcal{T}, H(X_{\mathcal{T}}^{(\pi[m])} | W_{\pi[1:m]},\mathcal{Q}) = H(X_{\mathcal{T}}^{(\pi[m])}, X_{\bar{\mathcal{T}}}^{(\pi[m])}| W_{\pi[1:m]},\mathcal{Q}) \quad (\mathrm{Due\ to\ (\ref{ind4})}), \label{dep} \\
        &H(X_{\bar{\mathcal{T}}}^{(\pi[\bar{m}])}|X_{\mathcal{T}}^{(\pi[m])},X_{\bar{\mathcal{T}}}^{(\pi[m])},\mathcal{Q}) = H(X_{\bar{\mathcal{T}}}^{(\pi[\bar{m}])}|X_{\mathcal{T}}^{(\pi[m])},W_{\pi[m]},\mathcal{Q}) \quad (\mathrm{Due\ to\ (\ref{con0})}). \label{ind7}
    \end{align}
    The following equations are the summary of the above equations:
    \begin{align}
        &\forall i,j\in [1:S], i\neq j ,\forall m \in [1:M], \forall \bar{\mathcal{I}}\subset [1:M]\setminus\{m\}, H(X_{j}^{(m)}| Q_{1:S}^{(m)}, W_{\bar{\mathcal{I}}},X_{i}^{(m)}) = H(X_{j}^{(m)}| Q_{1:S}^{(m)}, W_{\bar{\mathcal{I}}}) \label{sum1},
        \\&\forall \mathcal{T}, H(X_{\mathcal{T}}^{(m)} | W_{m},Q_{1:S}^{(m)}) = H(X_{\mathcal{T}}^{(m)}, X_{\bar{\mathcal{T}}}^{(m)}| W_{m},Q_{1:S}^{(m)}) \label{sum2}.
    \end{align}
\begin{proof}
When (\ref{ind5}) holds, the following equations also hold:
\begin{align}
    &\forall i,j \in [1:S], i \neq j ,I(X_i^{(\pi[m])};X_j^{(\pi[m])} | W_{\pi[1:m-1]},\mathcal{Q})=0
   \\& \Leftrightarrow  \forall j \in [1:S], \forall m \in [1:M], \forall \bar{\mathcal{I}}\subset [1:M]\setminus\{m\}, \notag
   \\& \quad H(X_{j}^{(m)}| Q_{1:S}^{(m)}, W_{\bar{\mathcal{I}}},X_{i}^{(m)}) = H(X_{j}^{(m)}| Q_{1:S}^{(m)}, W_{\bar{\mathcal{I}}}) \label{pi}.
\end{align}
(\ref{pi}) is because we can set $\pi[m]=m, \pi[1:m-1]=\bar{\mathcal{I}}$.
Due to (\ref{pi}), the equations (\ref{han2})-(\ref{ind6}) are all satisfied. Next, for (\ref{dep}), the following equations hold:
\begin{align}
    &\forall \mathcal{T}, H(X_{\mathcal{T}}^{(\pi[m])} | W_{\pi[1:m]},\mathcal{Q}) = H(X_{\mathcal{T}}^{(\pi[m])}, X_{\bar{\mathcal{T}}}^{(\pi[m])}| W_{\pi[1:m]},\mathcal{Q}) 
    \\& \Leftrightarrow \forall \mathcal{T}, H(X_{\mathcal{T}}^{(m)} | W_{m},Q_{1:S}^{(m)}) = H(X_{\mathcal{T}}^{(m)}, X_{\bar{\mathcal{T}}}^{(m)}| W_{m},Q_{1:S}^{(m)}) \label{pi2}.
\end{align}
(\ref{pi2}) is because we can set $\pi[1]=m$.
Therefore, (\ref{ind7}) is satisfied as follows:
\begin{align}
   &H(X_{\bar{\mathcal{T}}}^{(\pi[\bar{m}])}|X_{\mathcal{T}}^{(\pi[m])},X_{\bar{\mathcal{T}}}^{(\pi[m])},\mathcal{Q}) 
\\ & =H(X_{\bar{\mathcal{T}}}^{(\pi[\bar{m}])}|X_{\mathcal{T}}^{(\pi[m])},X_{\bar{\mathcal{T}}}^{(\pi[m])},W_{\pi[m]} ,\mathcal{Q}) 
\\ & =H(X_{\bar{\mathcal{T}}}^{(\pi[\bar{m}])},X_{\mathcal{T}}^{(\pi[m])},X_{\bar{\mathcal{T}}}^{(\pi[m])}|W_{\pi[m]} ,\mathcal{Q})-H(X_{\mathcal{T}}^{(\pi[m])}, X_{\bar{\mathcal{T}}}^{(\pi[m])}| W_{\pi[m]},Q_{1:S}^{(\pi[m])})
\\ & =H(X_{\bar{\mathcal{T}}}^{(\pi[\bar{m}])},X_{\mathcal{T}}^{(\pi[m])}|W_{\pi[m]} ,\mathcal{Q})-H(X_{\mathcal{T}}^{(\pi[m])}| W_{\pi[m]},Q_{1:S}^{(\pi[m])})
\\ & =H(X_{\bar{\mathcal{T}}}^{(\pi[\bar{m}])} | X_{\mathcal{T}}^{(\pi[m])}, W_{\pi[m]} ,\mathcal{Q}) .
\end{align}
\end{proof}
Let $Q_j^{(m)}=(Q_j^{(m)}[1],\dots,Q_j^{(m)}[M])$, the necessary and sufficient conditions for achieving the capacity can be rewritten as follows:
    \begin{align}
    &\forall j \in [1:S], \forall m \in [1:M], \forall \bar{\mathcal{I}}\subset [1:M]\setminus \{m\}, \mathrm{rank}\ Q_{1:S}^{(m)}[\mathcal{I}]=\sum_{j=1}^{S} \mathrm{rank}\ Q_j^{(m)}[\mathcal{I}], \label{con1}
    \\&\forall \mathcal{T} \subset [1:S], \forall m \in [1:M], \mathrm{rank}\ Q_{1:S}^{(m)}[\bar{m}]=\sum_{j\in \mathcal{T}} \mathrm{rank}\ Q_{j}^{(m)}[\bar{m}] . \label{con2}
    \end{align}
\begin{proof}
    For (\ref{con1}), the following equations hold:
   \begin{align}
    &H(X_{j}^{(m)}| Q_{1:S}^{(m)}, W_{\bar{\mathcal{I}}}) = \mathrm{rank}\ Q_{j}^{(m)} [\mathcal{I}] \log S ,
    \\ & H(X_{j}^{(m)}| Q_{1:S}^{(m)}, W_{\bar{\mathcal{I}}},X_{j}^{(m)}) = \mathrm{rank}\ Q_{\{i,j\}}^{(m)} [\mathcal{I}] \log S  - \mathrm{rank}\ Q_{j}^{(m)} [\mathcal{I}]  \log S.
   \end{align}
   Due to (\ref{sum1}), the following equations hold:
   \begin{align}
    &\forall i,j \in [1:S], \mathrm{rank}\ Q_{\{i,j\}}^{(m)} [\mathcal{I}] \log S - \mathrm{rank}\ Q_{j}^{(m)} [\mathcal{I}] \log S = \mathrm{rank}\ Q_{i}^{(m)} [\mathcal{I}] \log S
    \\ & \Leftrightarrow \mathrm{rank}\ Q_{1:S}^{(m)}[\mathcal{I}] \log S=\sum_{j=1}^{S} \mathrm{rank}\ Q_j^{(m)}[\mathcal{I}] \log S
    \\& \Leftrightarrow  \mathrm{rank}\ Q_{1:S}^{(m)}[\mathcal{I}] =\sum_{j=1}^{S} \mathrm{rank}\ Q_j^{(m)}[\mathcal{I}].
   \end{align}
   Next, for (\ref{con2}), the following equations hold:
   \begin{align}
    &H(X_{\mathcal{T}}^{(m)} | W_{m},Q_{1:S}^{(m)}) = \sum_{j\in \mathcal{T}} \mathrm{rank}\ Q_{j}^{(m)}[\bar{m}]\log S ,
    \\ &H(X_{\mathcal{T}}^{(m)}, X_{\bar{\mathcal{T}}}^{(m)}| W_{m},Q_{1:S}^{(m)}) = \mathrm{rank}\ Q_{1:S}^{(m)}[\bar{m}] \log S.
   \end{align}
   Therefore, we can rewrite (\ref{sum2}) as follows:
   \begin{align}
    &H(X_{\mathcal{T}}^{(m)} | W_{m},Q_{1:S}^{(m)}) = H(X_{\mathcal{T}}^{(m)}, X_{\bar{\mathcal{T}}}^{(m)}| W_{m},Q_{1:S}^{(m)})
    \\&\Leftrightarrow \mathrm{rank}\ \sum_{j\in \mathcal{T}} \mathrm{rank}\ Q_{j}^{(m)}[\bar{m}]\log S =  \mathrm{rank}\ Q_{1:S}^{(m)}[\bar{m}] \log S
    \\&\Leftrightarrow \mathrm{rank}\ \sum_{j\in \mathcal{T}} \mathrm{rank}\ Q_{j}^{(m)}[\bar{m}] = \mathrm{rank}\ Q_{1:S}^{(m)}[\bar{m}] .
   \end{align}
\end{proof}

\section{Conclusion}
In this study, we derived the necessary and sufficient conditions for privacy, correctness, and capacity achievability for standard and colluding PIR.
The conditions for standard PIR are summarized in Table \ref{sum}, and the conditions for colluding PIR are summarized in Table \ref{colludingterms}.  
\begin{table}[H]
    \begin{center}
    \caption{Conditions for Standard PIR}  \label{sum}
    \begin{tabular}{|c||c|}
       Properties   &  Conditions \\ \Hline{1.5pt}
    Correctness &  $\exists D_m,  D_m \begin{pmatrix}Q_1^{(m)}\\\vdots \\ Q_S^{(m)}\end{pmatrix} = 
    \begin{pmatrix}
      \mathbf{O}_{(m-1)L_w \times L_w}\\
      \mathbf{E}\\
      \mathbf{O}_{(M-m)L_w \times L_w}
     \end{pmatrix}^{\top}$
    \\ \hline
    Privacy & $\forall j \in [1:S], |\{Q_{1:S}^{(\theta)} | Q_{j}^{(\theta)} = q_j\}|=M|\{Q_{1:S}^{(m)} | Q_{j}^{(m)} = q_j\}|$\\ \hline 
    Capacity & 
     $\forall j \in [1:S],\forall m \in [1:M], \mathrm{rank}\ Q_{1:S}^{(m)}[\bar{m}]=\mathrm{rank}\ Q_j^{(m)}[\bar{m}],$ \\ 
     & $\forall \bar{\mathcal{I}}\subset [1:M]\setminus\{m\}, \mathrm{rank}\ Q_{1:S}^{(m)}[\mathcal{I}]=\sum_{j=1}^{S} \mathrm{rank}\ Q_j^{(m)}[\mathcal{I}]$\\ \Hline{1.5pt}
    \end{tabular}
    \end{center}
    \end{table}
    \begin{table}[H]
      \begin{center}
  \caption{Conditions for Colluding PIR}  \label{colludingterms}
  \begin{tabular}{|c||c|}
        Properties    & Conditions \\ \Hline{1.5pt}
  Correctness & 
  $\exists D_m,  D_m \begin{pmatrix}Q_1^{(m)}\\\vdots \\ Q_S^{(m)}\end{pmatrix} = 
  \begin{pmatrix}
    \mathbf{O}_{(m-1)L_w \times L_w}\\
    \mathbf{E}\\
    \mathbf{O}_{(M-m)L_w \times L_w}
   \end{pmatrix}^{\top}$
  \\ \hline
  Privacy & $\forall \mathcal{T}\subset [1:S],|\{Q_{1:S}^{(\theta)} | Q_{\mathcal{T}}^{(\theta)} = q_{\mathcal{T}}\}|=M|\{Q_{1:S}^{(m)} | Q_{\mathcal{T}}^{(m)} = q_{\mathcal{T}}\}|$\\ \hline 
  Capacity &  
  $\forall \mathcal{T}\subset [1:S],|\mathcal{T}|=T, \forall m \in [1:M], \mathrm{rank}\ Q_{1:S}^{(m)}[\bar{m}]=\sum_{j\in \mathcal{T}} \mathrm{rank}\ Q_{j}^{(m)}[\bar{m}]$,\\
  &$\forall \bar{\mathcal{I}}\subset [1:M]\setminus\{m\}, \mathrm{rank}\ Q_{1:S}^{(m)}[\mathcal{I}]=\sum_{j=1}^{S} \mathrm{rank}\ Q_j^{(m)}[\mathcal{I}]$\\ \Hline{1.5pt}
 \end{tabular}
  \end{center}
  \end{table}

\bibliography{bib} 

\begin{thebibliography}{10}

\bibitem{7889028}
Hua Sun and Syed~Ali Jafar.
\newblock The capacity of private information retrieval.
\newblock {\em IEEE Transactions on Information Theory}, Vol.~63, No.~7, pp.
  4075--4088, 2017.

\bibitem{sun2018capacity}
Hua Sun and Syed~A. Jafar.
\newblock On the capacity of locally decodable codes, 2018.

\bibitem{8720234}
Chao Tian, Hua Sun, and Jun Chen.
\newblock Capacity-achieving private information retrieval codes with optimal
  message size and upload cost.
\newblock {\em IEEE Transactions on Information Theory}, Vol.~65, No.~11, pp.
  7613--7627, 2019.

\bibitem{8119895}
Hua Sun and Syed~Ali Jafar.
\newblock The capacity of robust private information retrieval with colluding
  databases.
\newblock {\em IEEE Transactions on Information Theory}, Vol.~64, No.~4, pp.
  2361--2370, 2018.

\bibitem{8457293}
Karim Banawan and Sennur Ulukus.
\newblock The capacity of private information retrieval from byzantine and
  colluding databases.
\newblock {\em IEEE Transactions on Information Theory}, Vol.~65, No.~2, pp.
  1206--1219, 2019.

\bibitem{9145637}
Zhuqing Jia and Syed~Ali Jafar.
\newblock On the asymptotic capacity of x-secure t-private information
  retrieval with graph-based replicated storage.
\newblock {\em IEEE Transactions on Information Theory}, Vol.~66, No.~10, pp.
  6280--6296, 2020.

\bibitem{abs-1710-01190}
Qiwen Wang and Mikael Skoglund.
\newblock Secure private information retrieval from colluding databases with
  eavesdroppers.
\newblock {\em CoRR}, Vol. abs/1710.01190, , 2017.

\bibitem{9023400}
Zhen Chen, Zhiying Wang, and Syed~Ali Jafar.
\newblock The capacity of t-private information retrieval with private side
  information.
\newblock {\em IEEE Transactions on Information Theory}, Vol.~66, No.~8, pp.
  4761--4773, 2020.

\bibitem{8598994}
Razane Tajeddine, Oliver~W. Gnilke, David Karpuk, Ragnar Freij-Hollanti, and
  Camilla Hollanti.
\newblock Private information retrieval from coded storage systems with
  colluding, byzantine, and unresponsive servers.
\newblock {\em IEEE Transactions on Information Theory}, Vol.~65, No.~6, pp.
  3898--3906, 2019.

\bibitem{8509632}
Qiwen Wang and Mikael Skoglund.
\newblock On pir and symmetric pir from colluding databases with adversaries
  and eavesdroppers.
\newblock {\em IEEE Transactions on Information Theory}, Vol.~65, No.~5, pp.
  3183--3197, 2019.

\bibitem{8437880}
Jingke Xu and Zhifang Zhang.
\newblock Building capacity-achieving pir schemes with optimal
  sub-packetization over small fields.
\newblock In {\em 2018 IEEE International Symposium on Information Theory
  (ISIT)}, pp. 1749--1753, 2018.

\bibitem{8552404}
Zhifang Zhang and Jingke Xu.
\newblock The optimal sub-packetization of linear capacity-achieving pir
  schemes with colluding servers.
\newblock {\em IEEE Transactions on Information Theory}, Vol.~65, No.~5, pp.
  2723--2735, 2019.

\end{thebibliography}
\bibliographystyle{junsrt} 
\end{document}